\documentclass[journal,twocolumn]{IEEEtran}
\newcommand{\bx}{{\bf x}}
\newcommand{\by}{{\bf y}}
\newcommand{\be}{{\bf e}}
\newcommand{\ba}{{\bf a}}
\newcommand{\bz}{{\bf z}}
\newcommand{\bs}{{\bf s}}
\newcommand{\bba}{{\bf A}}
\newcommand{\bbd}{{\bf D}}
\newcommand{\bbx}{{\bf X}}
\newcommand{\CC}{{\mathbb{C}}}
 \usepackage{amsmath,amssymb,amsthm,color,graphicx}
\usepackage{algorithm,algorithmic}

\newtheorem{theorem}{Theorem}

\pdfoutput=1

\begin{document}
\title{Sparse Phase Retrieval from Short-Time Fourier Measurements}

\author{\IEEEauthorblockN{Yonina~C.~Eldar\IEEEauthorrefmark{1},~\IEEEmembership{Fellow,~IEEE}, Pavel Sidorenko\IEEEauthorrefmark{2}, Dustin G.\ Mixon\IEEEauthorrefmark{3}, Shaby Barel\IEEEauthorrefmark{1} and Oren Cohen\IEEEauthorrefmark{2}}
\thanks{\IEEEauthorrefmark{1}Dept. of Electrical Engineering, Technion -- Israel Institute of Technology, Haifa, Israel. \IEEEauthorrefmark{2}Department of Physics and Solid State
Institute, Technion. \IEEEauthorrefmark{3}Department of Mathematics and Statistics, Air Force Institute of Technology, Wright-Patterson AFB, Ohio. YE was supported by the SRC and by the Intel Collaborative Research Institute for Computational Intelligence. DGM was supported by NSF Grant No.\ DMS-1321779. This work was also supported by ICore - The Israeli
Excellence Center ``circle of light''.}}

\maketitle

\begin{abstract}
We consider the classical 1D phase retrieval problem. In order to overcome the difficulties associated with phase retrieval from measurements of the Fourier magnitude, we
treat recovery from the magnitude of the short-time Fourier transform (STFT).
We first show that the redundancy offered by the STFT enables unique recovery for arbitrary nonvanishing inputs, under mild conditions. An efficient algorithm for recovery of a sparse input from the STFT magnitude is then suggested, based on an adaptation of the recently proposed GESPAR algorithm. We demonstrate through simulations that using the STFT leads to improved performance over recovery from the oversampled Fourier magnitude with the same number of measurements.
\end{abstract}


\section{Introduction}

The problem of phase retrieval, namely recovering a signal from its Fourier transform magnitude, occurs in many fields of science and engineering, including electron microscopy, crystallography, optical imaging such as
coherent diffraction imaging (CDI), and diagnostics of ultra-short laser
pulses \cite{S14r,optics,crystallography}.
Here, we consider the 1D discrete phase retrieval problem. It is well known that there are many 1D
signals $\bx \in \CC^N$ with the same Fourier magnitude. This is true even if we eliminate trivial equivalences, such as a global phase shift, conjugate inversion and spatial shift, and even when the support
of the signal is bounded within a known range \cite{H64}.

One approach to try and overcome the non-uniqueness, is to exploit prior knowledge on
$\bx$. A popular prior that has been used extensively in signal processing is that $\bx$ is sparse, i.e. it contains only a small number $k$ of nonzero elements in an appropriate basis,
with $k \ll N$ \cite{EK12}. In this case, it has been shown
that there is a unique $\bx$ consistent with a given Fourier transform magnitude of length $M \geq 2N-1$
as long as $k \neq 6$ and the autocorrelation sequence of $\bx$
is collision free \cite{Ranieri2013,Kishore:12}. In \cite{OE13} uniqueness is established under these conditions when  $M$ is a prime number satisfying $M \geq k^2-k+1$.

Even when there is a unique solution to the phase retrieval problem, it is not clear how to find it using efficient and robust algorithms. The most popular techniques are based on alternating projections \cite{gerchberg1972practical,fienup,convex}, pioneered by Gerchberg and Saxton \cite{gerchberg1972practical}
and extended by Fienup \cite{fienup}. These methods generally require precise prior information on the signal such as knowledge
of the support set and often converge to erroneous results.
More recently, phase retrieval has been treated using semidefinite programming (SDP) and low-rank matrix recovery ideas \cite{CESV12,SESE11}.
 Both the SDP \cite{SESE11,JOH12,OYDS12,Mallat} and Fienup recovery methods \cite{M07,sparseFienup} have been extended to phase retrieval of sparse inputs.
A greedy approach to phase retrieval which can be applied in both the sparse and non-sparse setting was developed in \cite{SBE14} under the name of GESPAR: GrEedy Sparse PhAse Retrieval.
This method is motivated by the local search-type techniques of \cite{BE12}, and was shown to lead to improved performance over both SDP and Fienup-based algorithms, with low computational complexity.

To allow for recovery of a broader set of inputs, here we consider recovery from the magnitude of the short-time Fourier transform (STFT) \cite{NQL83}.
Phase retrieval from the STFT magnitude
has been used in several signal processing applications, for example in speech and audio processing \cite{NQL83,GL84}.
It has also been applied extensively in optics. One example is in
frequency resolved optical gating (FROG) or XFROG which are used for
characterizing ultra-short laser pulses by optically producing the STFT magnitude of the measured pulse \cite{T00,K08}.
% In FROG the pulse itself (or a function of the pulse) is used to gate the measured signal while in XFROG gating is performed by a fixed known window.
Another example is ptychographical CDI \cite{GF08}.
Both Fienup-type methods \cite{GL84,K08} and SDP approaches have been extended to recovery from the STFT magnitude \cite{SSJ12}.

Here, we first show that the redundancy offered by the STFT enables unique recovery for arbitrary inputs $\bx$ that are nonvanishing, under mild conditions. We then suggest an efficient algorithm for recovery of a sparse input from the STFT magnitude, based on an adaptation of GESPAR to STFT measurements. We compare this approach to recovery from the Fourier magnitude with oversampling such that the number of measurements is the same in both settings. Our simulations demonstrate that applying GESPAR allows one to recover sparse signals from the STFT magnitude even in cases where recovery from the Fourier magnitude fails.

\section{Phase Retrieval from STFT Measurements}

\subsection{The STFT}

Consider a length-$N$ signal $x[n]$ defined on $0 \leq n \leq N-1$.
The short-time Fourier transform (STFT) of $x[n]$ is defined as
\begin{equation}
\label{eq:stft}
X_g(m,k)=\sum_{n=0}^{N-1} x[n]g[mL-n]e^{-j2\pi k n/N},
\end{equation}
where $L$ is a given parameter and
 $g[n]$ is an appropriate window. The parameter $L$ denotes the
separation in time between adjacent sections.
We assume throughout that $g[n]$ is  periodically extended over the boundaries in \eqref{eq:stft}.

To invert the STFT operation we consider the inverse discrete Fourier transform (DFT) of $X_g(m,k)$:
\begin{equation}
\label{eq:istft}
x_g(m,n)=\frac{1}{N}\sum_{k=0}^{N-1} X_g(m,k) e^{j2\pi k n/N}.
\end{equation}
We can then obtain $x[n]$ by
\begin{equation}
\label{eq:istftx}
x[n]=\frac{\sum_{m} x_g(m,n)\overline{g[mL-n]}}
{\sum_{m } |g[mL-n]|^2}.
\end{equation}

\subsection{Phase Retrieval}

Suppose now that we are given only the magnitude $|X_g(m,k)|$ of the STFT. The question we would like to address is whether we can recover $x[n]$.
The following theorem provides a first step towards phase retrieval from the STFT.
We define the length $W$ of the window $g$ to be the size of the smallest interval $[a,b]\subseteq\{0,\ldots,N-1\}$ such that the support of $g$ is contained in $[a,b]$; since $g$ is $N$-periodic, the endpoints $a$ and $b$ are to be interpreted modulo $N$.

The theorem below shares some similarities with a related result from \cite{NQL83}. The main difference between the two settings is that in our definition of the STFT we consider a periodically extended window, and a sampled Fourier transform.
\begin{theorem}
Let $g[n]$ be a window of length $W\geq 2$ and consider the STFT defined by \eqref{eq:stft} with $L=1$.
Then $|X_g(m,k)|^2$ uniquely determines every $x[n]$ with nonvanishing entries (up to a global phase factor) provided
\begin{itemize}
\item[(i)] the length-$N$ DFT of $v[n]=|g[n]|^2$ is nonvanishing,
\item[(ii)] $N\geq 2W-1$, and
\item[(iii)] $N$ and $W-1$ are coprime.
\end{itemize}
In the special case in which $g[n]$ is the $N$-periodic extension of $1_{[0,W-1]}$, condition (i) is satisfied when $N$ and $W$ are coprime.
\end{theorem}

\begin{proof}
From Plancherel's theorem we have that
\begin{align}
\sum_{k=0}^{N-1}|X_g(m,k)|^2
&=N\sum_{n=0}^{N-1}|x[n]g[m-n]|^2\\
&=N\sum_{n=m-W+1}^m|x[n]|^2|g[m-n]|^2,
\end{align}
where the indices in the last sum are to be interpreted modulo $N$. Let $\by$ be the length-$N$ vector with $m$th element
$y_m=\sum_{k=0}^{N-1}|X_g(m,k)|^2$, and let $\bz$ be the length-$N$ vector with $m$th element $|x[m]|^2$.
Let $[a,a+W-1]$ denote the interval of size $W$ that contains the support of $g$.
It then follows that $\by=N \bba\bz$ where $\bba$ is an $N \times N$ circulant matrix whose $(a+W-1)$st row is given by
$\ba=[|g[a+W-1]|^2,\ldots,|g[a]|^2,0,\ldots,0]^T$.
Since circulant matrices are diagonalized by the DFT, $\bba$ is invertible if and only if the DFT of $\ba$ is nowhere vanishing, which is equivalent to (i).

When $g[n]$ is a square window, $\ba$ is given by $1_{[0,W-1]}$. The length-$N$ DFT will then vanish only if $\sin(\pi kW/N)$ is zero for some $k=1,\ldots,N-1$, which will not happen when
$N$ and $W$ are coprime.

At this point, we have determined $\bz$ from $\by$ so that we know $|x[n]|$ for all $n$.
It remains to find the relative phases between the entries of $x[n]$.
To this end, we will use the fact that if $y[n]$ is a length-$N$ sequence with DFT $Y[k]$, then $D[k]=|Y[k]|^2$ is the DFT of the circular autocorrelation:
\begin{equation}
\label{eq:corr}
r[n]=\sum_{p=0}^{N-1}y[p]\overline{y[p-n]},
\end{equation}
where $p-n$ should be interpreted modulo $N$.
In our case, for every $m$ we can define $Y[k]=X_g(m,k)$, and
$y[n]=x[n]g[m-n]$ to get
\begin{align}
r[n]
&=\frac{1}{N}\sum_{k=0}^{N-1}|X_g(m,k)|^2e^{j2\pi kn/N}\\
&=\sum_{p=0}^{N-1}x[p]g[m-p]\overline{x[p-n]g[m-p+n]}.
\end{align}

Consider $r[W-1]$.
Again, let $[a,a+W-1]$ denote the interval of size $W$ that contains the support of $g$.
Then by condition (ii), $g[m-p]\overline{g[m-p+W-1]}$ is nonzero if and only if $m-p=a$, meaning
\begin{equation}
r[W-1]
=x[m-a]g[a]\overline{x[m-a-W+1]g[a+W-1]}.
\end{equation}
We next divide by the known values $|x[m-a]|$, $|x[m-a-W+1]|$, $g[a]$ and $\overline{g[a+W-1]}$ to isolate the phase $\operatorname{arg}(x[m-a])-\operatorname{arg}(x[m-a-W+1])$; this division is possible since $x[n]$ is nowhere vanishing and $g[a]$ and $g[a+W-1]$ are nonzero by the definitions of $a$ and $W$.
Finally, assuming (iii), we can arbitrarily set $x[0]$ to be positive and propagate these relative phases to determine $x[n]$ up to a global phase.
\end{proof}

To evaluate our theorem, we note its shortcomings.
While it ensures that almost every signal $x[n]$ is uniquely determined by the modulus of its STFT (with general conditions on $g[n]$), the proof makes use of the fact that $L=1$.
If $L>1$, then the linear system $\by=N \bba\bz$ is underdetermined.
To be fair, this just establishes that our proof technique is insufficient in the $L>1$ regime; since our proof only makes use of $r[0]$ and $r[W-1]$ for $r[n]$ defined by \eqref{eq:corr}, the STFT has more redundancy to leverage, so that one should expect phase retrieval to be robust to downsampling the STFT.
We confirm this intuition with simulations which clearly demonstrate the feasibility of phase retrieval when $L>1$. We also show in simulations that the length of the Fourier transform in \eqref{eq:stft} can be chosen equal to the window length.

Next, while our theorem provides a guarantee for nonvanishing signals, it says nothing about sparse signals.
This is another artifact of our proof technique.
Based on intuitions from compressed sensing, one might expect to recover sparse signals, perhaps with even larger downsampling rates $L$.
However, as the following theorem establishes, there is a fundamental limit to how sparse a signal can be in the identity basis:
\begin{theorem}
Consider nonoverlapping intervals $[a_1,b_1],[a_2,b_2]\subseteq\{0,\ldots,N-1\}$ such that $a_2-b_1$ and $a_1-b_2$ (interpreted modulo $N$) are both $\geq W$, and take $x[n]$ supported on $[a_1,b_1]$ and $y[n]$ supported on $[a_2,b_2]$.
Then the modulus squared of the STFT of $x[n]+y[n]$ and of $x[n]-y[n]$ are equal for any $L$.
\end{theorem}

\begin{proof}
It suffices to prove the result for $L=1$.
Take $u[n]=x[n]+y[n]$ and $v[n]=x[n]-y[n]$.
For each $m$, put $m\in M_1$ if the support of $g[m-n]$ intersects $[a_1,b_1]$, $m\in M_2$ if it intersects $[a_2,b_2]$, and $m\in M_0$ if it intersects neither.
Depending on $\sigma=\pm1$, the following is an expression for either $U_g(m,k)$ or $V_g(m,k)$:
\begin{align}
&\sum_{n=0}^{N-1}(x[n]+\sigma y[n])g[m-n]e^{-j2\pi kn/N}\nonumber \\
&=\left\{
\begin{array}{cl}
X_g(m,k),&\mbox{if }m\in M_1\\
\sigma Y_g(m,k),&\mbox{if }m\in M_2\\
0,&\mbox{if }m\in M_0.
\end{array}
\right.
\end{align}
As such, $|U_g(m,k)|=|V_g(m,k)|$ for every $m$ and $k$.
\end{proof}

To provide some intuition to the difficulty encountered when the sparsity is in the identity basis, we use the
fact that time shifts preserve the Fourier magnitude to
construct a class of examples of sparse signals that cannot be recovered uniquely from the STFT magnitude. Consider, for example, a window $g[n]=1$ over its support, a signal $x[n]$ and
parameters $N,W,L$ such that $W$ and $N$ are both multiples of $L$. The signal $x[n]$ has a nonzero segment of length $L-r$ for some $1 \leq r \leq L-1$,
which is positioned within an interval of the form $[(a-1)L+1,aL]$ for some integer $a$.
This interval has at least $W-L$ zeros attached to it both on the left and on the right. An example is depicted in
 Fig.~\ref{fig:sparse}(a). Each such nonzero segment can be moved up to $r$ indices within the interval and can be multiplied by an arbitrary phase without affecting the STFT magnitude, as illustrated in the figure. In this example $N=64,W=16,L=4$ and $r=1$. The figure shows two signals that differ in their first nonzero segment. The segments are a simple shift by one tap and negative of each other. The resulting signals have the same STFT magnitude. The windows overlapping this segment are also illustrated --- for both signals, the Fourier transform magnitude of the windowed signal is the same in all cases. In contrast,  Fig.~\ref{fig:sparse}(b) shows a random sparse signal with $50$ nonzero values in the identity basis that is perfectly recovered using GESPAR.
\begin{figure} [t!]
%\centering
{\includegraphics[width=\columnwidth]{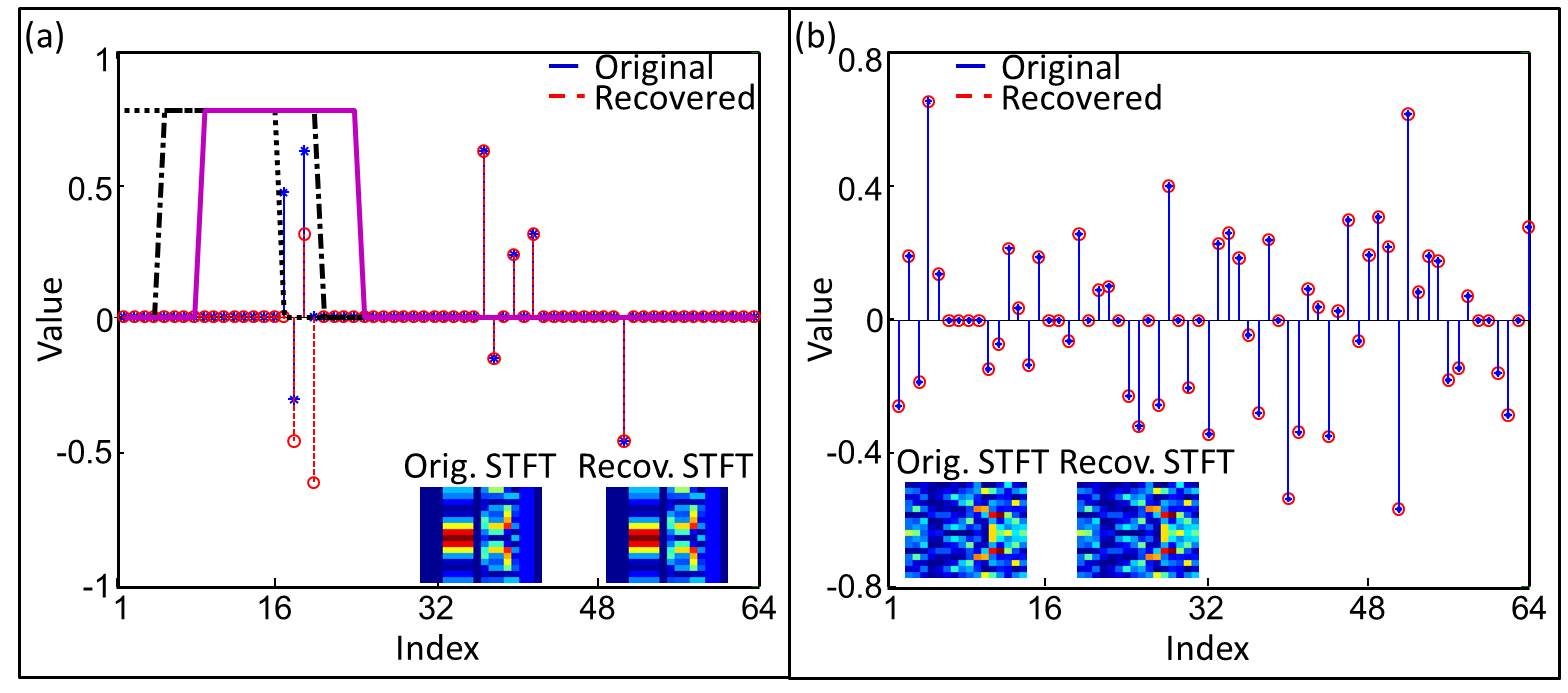}}
\hfil
\caption{(a) Ambiguities with sparse inputs (b) recovery of a sparse signal.}
\label{fig:sparse}
\end{figure}

\subsection{Sparse Phase Retrieval}

In practice, the performance of STFT phase retrieval can be improved by exploiting sparsity.
Specifically, suppose that $x[n]$ is sparse in a basis (or frame) represented by a matrix $\bbd$. This means that $\bx=\bbd\bs$ where $\bx$ is the length-$N$ vector with elements $x[n]$ and $\bs$ is a $k$-sparse vector such that $\|\bs\|_0=k \ll N$, where $\|\bs\|_0$ denotes the number of nonzero elements in $\bs$.
Our goal is to recover $\bx$, or equivalently $\bs$, given the STFT magnitude of $\bx$ by exploiting the sparsity of $\bs$.

To this end, we apply the GESPAR algorithm developed in \cite{SBE14} for phase retrieval of sparse vectors from general quadratic measurements.
In our case, $y[n]=|X_g(m,k)|^2$ can be written as a quadratic function of $\bs$ of the form:
\begin{equation}
\label{eq:qstft}
y[n]=|\be_n^* \textbf{F} \textbf{G}_m \textbf{D}\textbf{s}|^2,
\end{equation}
where $ \textbf{G}_m $ is an $N \times N$ diagonal matrix with diagonal elements $ g[mL-n] $ properly positioned along the diagonal,
$\textbf{F}$ is an $N \times N$ DFT matrix, and $\be_n$ is the $n$th column of the identity.

\section{Recovery Algorithms}

As noted in the introduction, alternating projection algorithms are the most popular techniques for phase retrieval. These methods have also been applied to recovery from the STFT magnitude leading to the well known Griffin-Lim algorithm (GLA) \cite{GL84}, summarized as Algorithm~\ref{alg:GLA}.
\begin{algorithm}[!t]
\caption{Griffin-Lim algorithm (GLA) \label{alg:GLA}}
\begin{algorithmic}
\STATE Input: Measurements $y(m,k)=|X_g(m,k)|^2$ and window $g[n]$, $k,n=0,\ldots,N-1, m=0,\ldots,M-1$
\STATE Output: Estimate $\hat{x}[n]$ of $x[n]$
\STATE Initialize: Choose a random input signal $x_0[n]$, $\ell = 0$
\WHILE{halting criterion false}
\STATE $\ell \leftarrow \ell+1$
\STATE $X_g^\ell(m,k)=\sum_{n=0}^{N-1} x_{\ell-1}[n]g[mL-n]e^{-j2\pi k n/N}$
\COMMENT{compute the STFT of $x_{\ell-1}[n]$}
\STATE $b(m,k)=\frac{X_g^\ell(m,k)}{|X_g^\ell(m,k)|}\sqrt{y(m,k)}$
\COMMENT{Keep phase and update magnitude}
\STATE Compute the inverse DFT $x_g^\ell(m,n)$ of $b(m,k)$ using \eqref{eq:istft}
\STATE $\hat{x}_\ell[n]=\frac{\sum_{m} x_g^\ell(m,n)\overline{g[mL-n]}}
{\sum_{m } |g[mL-n]|^2}$
\COMMENT{update signal estimate}
\ENDWHILE
\STATE return $\hat{x}[n] \leftarrow \hat{x}_\ell[n]$
\end{algorithmic}
\end{algorithm}
Note, that although we assumed that the DFT length in \eqref{eq:stft} is equal to the signal length $N$, in  practice, we can choose the DFT length, which we denote by $K$, to be any integer that is equal or larger than the window length $W$.

In optics, a slight variation of GLA is used, referred to as principal components generalized projections (PCGP) \cite{K08}.
It differs from Algorithm~\ref{alg:GLA} in the last step: Rather than computing the inverse DFT and updating the signal according to
\eqref{eq:istftx}, it updates the signal by performing a rank-one approximation to an appropriately formed matrix. Specifically, we first form  a matrix $\bbx^\ell$ whose $m$th row is
a circulant shift of $X_g^\ell(m,k)$ by $m$ indices. Next, a singular value decomposition (SVD) is performed on $\bbx^\ell$. The signal $\hat{x}_\ell[n]$ is then chosen as the left singular vector corresponding to the largest singular value.

In order to improve the recovery performance, we suggest exploiting sparsity in the signal input (when present). To this end, we apply the GESPAR algorithm to the quadratic measurements \eqref{eq:qstft}. More specifically, GESPAR is aimed at approximating the solution to the problem
\begin{equation}
\min_{\bs} \sum_n (y[n]-|\be_n^* \textbf{F} \textbf{G}_m \textbf{D}\textbf{s}|^2)^2
\quad \mbox{s.t. } \|\bs\|_0 \leq k.
\end{equation}
GESPAR is a 2-opt local search algorithm in which the signal support is iteratively updated
using the gradient of the objective function. A damped Gauss Newton method is used to minimize the objective over the support.

Below, we compare between two different applications of GESPAR using the same number of measurements. The first, STFT GESPAR, applies GESPAR to the squared-magnitude of the STFT. The second, referred to as PS GESPAR, uses GESPAR to recover $x[n]$ from its power spectrum, namely, the squared-magnitude of the Fourier transform.  We use an oversampled Fourier transform such that the number of measurements in PS GESPAR is the same as that in STFT GESPAR, which is equal to $P=MK$. Here $M$ is the number of windows and $K$ is the DFT length of each windowed signal.

\section{Simulations}
We now demonstrate the performance of STFT GESPAR via simulations. The sparsity dictionary is chosen as a random basis with iid standard normal variables, followed by normalization of the columns.
To generate the sparse inputs, for every sparsity level $k$ we choose $k$ locations for the nonzero components uniformly at random. The signal values over the selected support are then drawn from an iid standard normal distribution.  We compute the STFT using a square window of length $W=16$, so that $g[n]=1$ over its support.

In the first simulation, we examine the performance of GESPAR, GLA, and PCGP in recovering an unknown vector $\textbf{s}$ of length $N=64$ in \eqref{eq:qstft}, as a function of $L$ for $L=2,4,8,16$. The number of DFT points is chosen as $K=W=16$. The number of measurements in the STFT is then $P=512,256,128,64$. STFT GESPAR is used with a threshold $\tau=10^{-4}$ and maximum number of swaps $50000$; PCGP and GLA are run using 50 random initial points with 1000 maximal iterations.
For comparison with the Fourier transform approach, we show the performance of PS GESPAR with the same parameters
and oversampling factors of $8,4,2,1$.

In Fig.~\ref{fig:stftl} we plot the recovery probability as a function of sparsity by recording the percentage of successful recoveries in $100$ simulations, where success is declared if the normalized squared error is smaller than $10^{-2}$. Note, that when $L=16$ there is no redundancy in the STFT and therefore it is not surprising that there is no advantage to the STFT method. In all other cases for which $L<16$, the STFT introduces redundancy, which leads to improved performance over simply oversampling the DFT. It is also evident that GESPAR outperforms both GLA and PCGP.
\begin{figure} [h]
%\centering
{\includegraphics[width=\columnwidth]{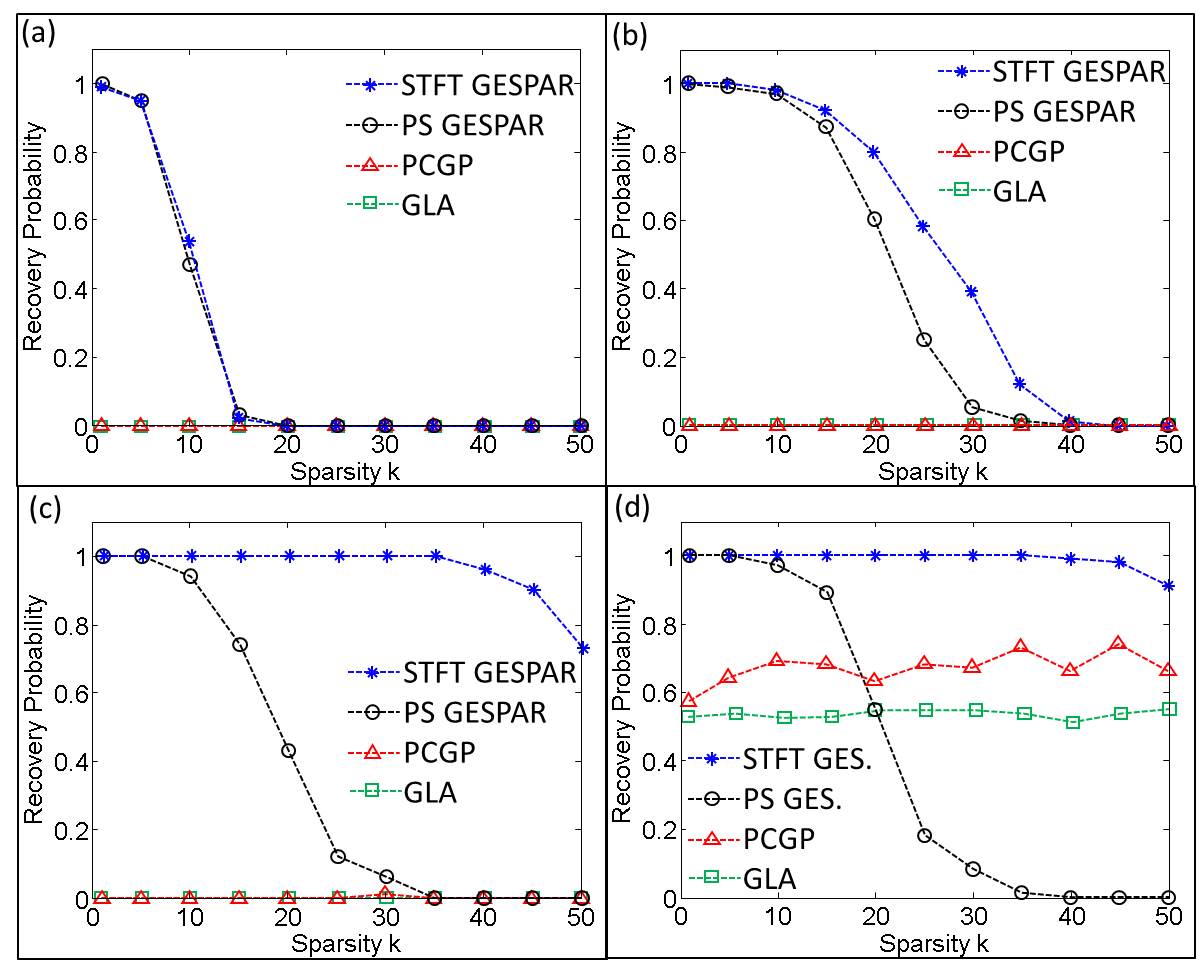}}
\hfil
\caption{Recovery probability vs sparsity $k$ for varying number of measurements (a) 64, (b) 128, (c) 256, (d) 512.}
\label{fig:stftl}
\end{figure}

In Fig.~\ref{fig:stftk} we consider the effect of noise and the DFT length on the normalized mean-squared error (NMSE). All parameters are the same as in Fig.~\ref{fig:stftl} besides $L$ which is set to $L=1$ and the signal length which is $N=32$. The DFT length is chosen as
  $K=2,4,8,16,32$. When $K$ is smaller than the window length $W=16$, we use a DFT of length $W$ and choose only the first $K$ measurements; i.e. we use only the $K$ low frequency measurements.
As expected, increasing the DFT length improves the recovery ability. It is also evident that the performance improves significantly when all Fourier components are measured, namely, when $K \geq 16$.
\begin{figure} [h]
\centering
{\includegraphics[width=\columnwidth]{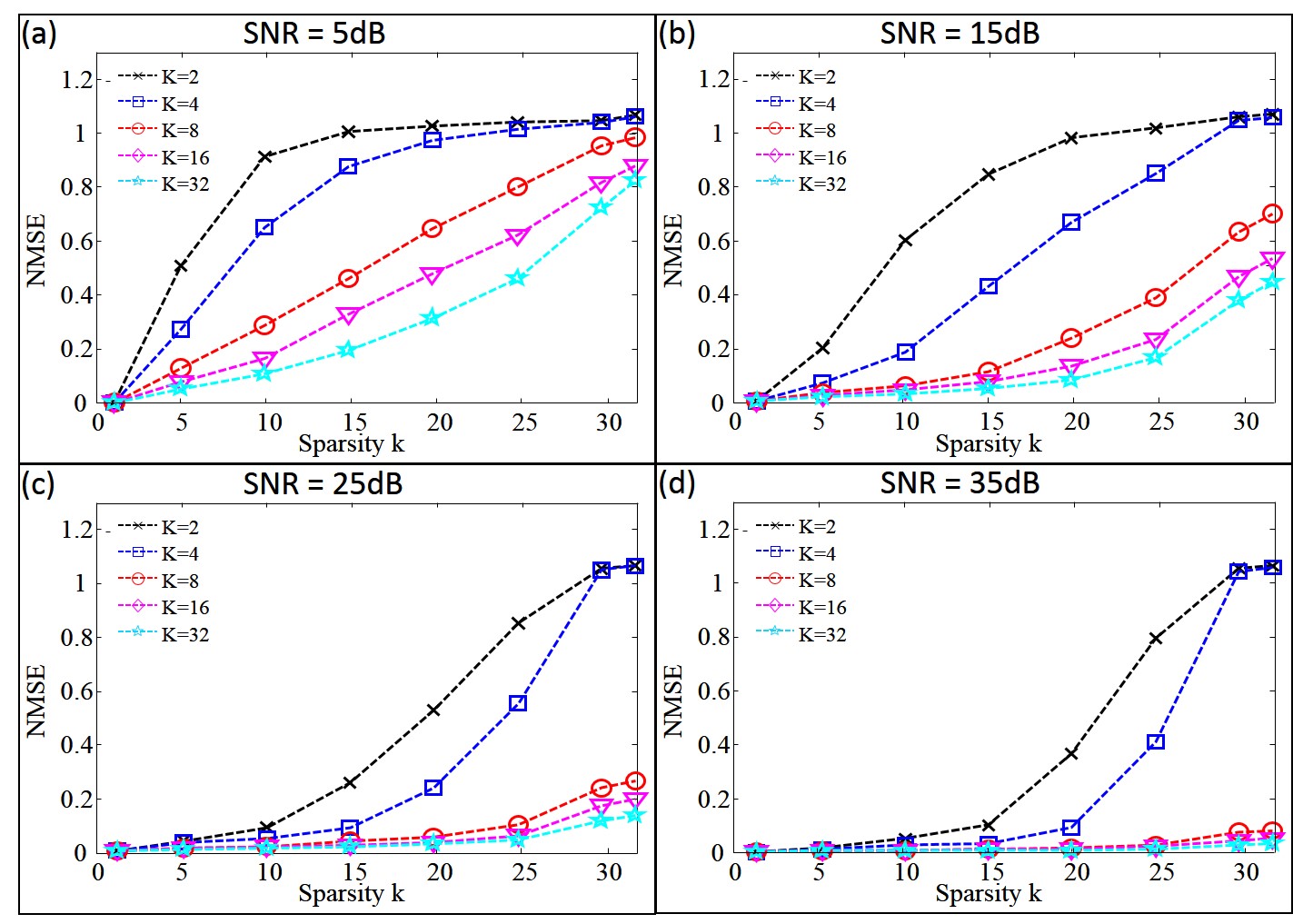}}
\hfil
\caption{Normalized MSE vs sparsity $k$ for varying number of frequencies $K= 2,4,8,16,32$ using STFT GESPAR and several SNR values (in [dB]) (a) 5, (b) 15, (c) 25, (d) 35.}
\label{fig:stftk}
\end{figure}

\section{Conclusion}

In this letter we suggest to improve the performance of phase retrieval methods by exploiting sparsity together with the STFT. We demonstrated that for the same number of measurements, using the STFT can lead to far better performance than using an oversampled DFT.
We also showed that GESPAR is able to exploit both the redundancy in the measurements and the sparsity of the input, leading to high probability of recovery as long as sufficient redundancy is introduced into the measurement process. There are many interesting theoretical directions remained to be explored in this context. In particular, it is important to understand under what conditions on the sparsity level and sparsity basis one can recover sparse inputs from the STFT, and how large $L$ can be made while still ensuring recovery for generic inputs.

\bibliographystyle{IEEEtrans}
\bibliography{STFT}

\end{document}